\documentclass[runninghead]{llncs}

\usepackage{graphicx}
\usepackage{environ}
\usepackage{booktabs} 
\usepackage{latexsym,graphicx,amsfonts,amsmath}
\usepackage{paralist}
\usepackage{multirow}
\usepackage{dirtytalk}
\usepackage[stable]{footmisc}
\usepackage{color}

\newtheorem{coro}{Corollary}
\newtheorem{obse}{Observation}

\newtheorem{defi}{Definition}

\newcommand{\VV}{{\mathcal V}}
\newcommand{\NP}{{\textrm{NP}}}
\newcommand{\PP}{{\textrm{P}}}
\newcommand{\head}[1]
\
{\markright{\hbox to 0pt{\vtop to 0pt{\hbox{}\vskip 3mm \hrule
				width  \textwidth \vss} \hss}{\sc ##1}}}

\begin{document}

\title{Computational Complexity Characterization of Protecting Elections from Bribery \footnote{A 2 page extended abstract has been published at the Proceedings of the 17th International Conference on Autonomous Agents and MultiAgent Systems (AAMAS'18)}}

%
%
\author{Lin Chen\inst{1} \and
Ahmed Sunny\inst{1}
Lei Xu\inst{2} \and
Shouhuai Xu\inst{3} \and
Zhimin Gao\inst{4} \and
Yang Lu\inst{5} \and
Weidong Shi\inst{5} \and 
Nolan Shah\inst{6} }
\authorrunning{F. Author et al.}
%
\institute{Texas Tech University, 2500 Broadway, Lubbock, TX 79409, USA \email{lin.chen@ttu.edu}  \email{ahmed.sunny@ttu.edu} \and
University of Texas Rio Grande Valley,1201 W University Dr, Edinburg, TX 78539, USA \email{xuleimath@gmail.com}\and 
University of Texas San Antonio,1 UTSA Circle, San Antonio, TX 78249, USA  \email{shouhuai.xu@utsa.edu} \and Auburn University at Montgomery,7430 East Dr, Montgomery, AL 36117, USA  \email{mtion@masn.com} \and
University of Houston, 4800 Calhoun Rd, Houston, TX 77004, USA   \email{ylu17@central.edu} \email{wshi3@uh.edu} \and
Amazon Web Services, Seatle, USA
\email{nolan@0x9b.com}}

%
\maketitle              

\begin{abstract}
	The bribery problem in election has received considerable attention in the literature, upon which various algorithmic and complexity results have been obtained. It is thus natural to ask whether we can protect an election from potential bribery. We assume that the protector can protect a voter with some cost (e.g., by isolating the voter from potential bribers). A protected voter cannot be bribed. Under this setting, we consider the following bi-level decision problem: Is it possible for the protector to protect a proper subset of voters such that no briber with a fixed budget on bribery can alter the election result? The goal of this paper is to give a full picture on the complexity of protection problems. We give an extensive study on the protection problem and provide algorithmic and complexity results. Comparing our results with that on the bribery problems, we observe that the protection problem is in general significantly harder. Indeed, it becomes $\Sigma_2^p$-complete even for very restricted special cases, while most bribery problems lie in \NP. However, it is not necessarily the case that the protection problem is always harder. Some of the protection problems can still be solved in polynomial time, while some of them remain as hard as the bribery problem under the same setting.
	

\end{abstract}

\section{Introduction}

In an election, there are a set of candidates and a set of voters. Each voter has a {\em preference list} of candidates. Given these preference lists, a winner is determined based on some {voting rule}, examples of which will be elaborated later.

In the context of election, the {\em bribery problem} has received considerable attention (see, for example, \cite{bredereck2016prices,erdelyi2020complexity,faliszewski2009hard,kaczmarczyk2019algorithms,knop2017voting}).
In this problem, there is an attacker who attempts to manipulate the election by bribing some voters, who will then report preference lists of the attacker's choice (rather than the voters' own preference lists). Each voter has a price for being bribed, and the attacker has an attack budget for bribing voters.
There are two kinds of  attackers: {\em constructive} vs. {\em destructive}.
A {\em constructive} attacker attempts to make its designated candidate win an election, whereas the designated candidate would not win the election should there be no attacker. In contrast, a {\em destructive} attacker attempts to make its designated candidate lose the election, whereas the designated candidate would win the election should there be no attacker. 
The research question is: Given an attack budget for bribing, whether or not a (constructive or destructive) attacker can achieve its goal?

In this paper, we initiate the study of a new problem, called the {\em protection problem},
which extends the bribery problem as follows. There are also a set of candidates, a set of voters, and a bribery attacker. Each voter also has a {\em preference list} of candidates. There is also a { voting rule} according to which a winner is determined.
Going beyond the bribery problem, the protection problem further considers a defender, who aims to protect elections from bribery.
More specifically, the defender is given a defense budget and can use the defense budget to award some of the voters so that they cannot be bribed by the attacker anymore.
This leads to an interesting problem: {\em Given a defense budget, is it possible to protect an election from an attacker with a given attack budget for bribing voters (i.e., assuring that the attacker cannot achieve its goal)?}




\smallskip
\noindent\textbf{Our contributions.}
We introduce the problem of protecting elections from bribery, namely the protection problem.
Given a defense budget for rewarding some of the voters and an attack budget for bribing some of the rest voters, the protection problem asks whether or not the defender can protect the election.
We investigate the protection problem against the aforementioned two kinds of bribery attackers: constructive vs. destructive. 


We present a characterization on the computational complexity of the protection problem (summarized in Table \ref{table:1} in {Section \ref{sec:discussion}}).
The characterization is primarily concerning the voting rule of $r$-approval, which will be elaborated in Section \ref{sec:definition}.
%
%
At a high level, our results can be summarized as follows. (i) The {\em protection problem} is hard and might be much harder than the {\em bribery problem}.
For example, the protection problem is $\Sigma_2^p$-complete in most cases,
while the bribery problem is in \NP\, under the same settings.
(ii) The {\em destructive protection problem} (i.e., protecting elections against a destructive attacker) is {\em no} harder than the {\em constructive protection problem} (i.e., protecting elections against a constructive attacker) in all of the settings we considered. In particular, the destructive protection problem is $\Sigma_2^p$-hard only when the voters are weighted and have arbitrary prices, while the constructive protection problem is  $\Sigma_2^p$-hard even when the voters are unweighted and have the unit price.
(iii) Voter weights and prices have completely different effects on the computational complexity of the protection problem.
For example, the constructive protection problem is \textrm{coNP}-hard in one case
but is in \PP\, in another case.


\smallskip

\noindent\textbf{Related Work.}
The problem of protecting elections from attacks seemingly has not received the due attention.
Very recently, Yin et al. \cite{yin2016optimally} considered the problem of defending elections against an attacker who can delete (groups of) voters.
That is, the investigation is in the context of the {\em control problem}, where the attacker attempts to manipulate an election by adding or deleting some voters. 
The control problem has been extensively investigated (see, for example, \cite{chen2017elections,faliszewski2015weighted,faliszewski2016control,yin2016optimally}).
Although the control problem is related to the bribery problem, the means used by the attacker in the control problem (i.e., attacker adding or deleting some voters)
is different from the means used by the attacker in the bribery problem (i.e., attacker changing the preference lists of the bribed voters).
We investigate the protection problem, which is defined in the context of the bribery problem rather than the control problem.
That is, the problem we investigate is different from the problem investigated by Yin et al. \cite{yin2016optimally}.

The protection problem we study is inspired by the bribery problem.
Faliszewski et al. \cite{faliszewski2009hard} gave the first characterization on the complexity of the {\em bribery problem}, including some dichotomy theorems.
In the bribery problem, the attacker can pay a fixed, but voter-dependent, price to arbitrarily manipulate the preference list of a bribed voter. The complexity of the bribery problem under the scoring rule of $r$-approval or $r$-veto for small values of $r$ was addressed later by Lin \cite{lin2010complexity} and Bredereck and Talmon \cite{bredereck2016np}. There are also studies on measuring the bribery price in different ways (see, e.g.,~\cite{bredereck2016prices,dey2017frugal,kaczmarczyk2019algorithms}).



Technically, the protection problem is related to the {\em bi-level optimization} problem, especially the {\em bi-level knapsack} problem (\cite{caprara2014study,chen2013approximation,qiu2015improved,wang2010two}). 
In the bi-level knapsack problem, there is a leader and a follower. The leader makes a decision first (e.g., packing a subset of items into the knapsack), and then the follower solves an optimization problem given the leader's decision (e.g., finding the most profitable subset of items that have not been packed by the leader). The problem asks for the decision of the leader such that a certain objective function is optimized (e.g., minimizing the profit of the follower). The protection problem we study can be formulated as the bi-level problem by letting the defender award some voters who therefore cannot be bribed by the attacker anymore, and then the attacker bribes some of the remaining voters as an attempt to manipulate the election.

\section{Problem definition}
\label{sec:definition}
\noindent\textbf{Election model.}
Consider a set of $m$ candidates $\mathcal{C}=\{c_1,c_2,\ldots, c_m\}$ and a set of $n$ voters $\mathcal{V}=\{v_1,v_2,\ldots,v_n\}$. Each voter $v_j$ has a preference list of candidates, which is essentially a permutation of candidates,
denoted as $\tau_j$. The preference of $v_j$ is denoted by $(c_{\tau_j(1)},c_{\tau_j(2)},\ldots,c_{\tau_j(m)})$, meaning that $v_j$ prefers candidate $c_{\tau_j(z)}$ to $c_{\tau_j(z+1)}$, where $z=1,2,\ldots$. Since $\tau_j$ is a permutation over $\{1,2,\ldots,m\}$, we denote by $\tau_j^{-1}$ the inverse of $\tau_j$, meaning that $\tau_j^{-1}(i)$ is the position of candidate $c_i$ in vector $(c_{\tau_j(1)},c_{\tau_j(2)},\ldots,c_{\tau_j(m)})$.


\smallskip\noindent\textbf{Voting rules.}
In this paper, we focus on the scoring rule (or scoring protocol) that maps a preference list to a $m$-vector $\alpha=(\alpha_1,\alpha_2,\ldots,\alpha_m)$, where $\alpha_i\in\mathbb{N}$ is the score assigned to the $i$-th candidate on the preference list of voter $v_j$ and $\alpha_1\ge \alpha_2\ge\ldots\ge \alpha_m$.  Given that $\tau_j$ is the preference list of $v_j$, candidate $c_{\tau_j(i)}$ receives a score of $\alpha_{i}$ from $v_j$. The {\em total score} of a candidate is the summation of the scores it received from the voters. The winner is the candidate that receives the highest total score. We focus on a {\em single-winner} election, meaning that only one winner is selected.
In the case of a tie, a random candidate with the highest total score is selected. However, our results remain valid for all-natural variation of selecting a single winner. 

We say a scoring rule is {\em non-trivial}, if $\alpha_1>\alpha_m$ (i.e., not all scores are the same).
There are many (non-trivial) scoring rules, including the popular {\em $r$-approval}, {\em plurality}, {\em veto}, {\em Borda count} and so on. In the case of $r$-approval, $\alpha=(\underbrace{1,1,\ldots,1}_{r},\underbrace{0,0,\ldots,0}_{m-r})$. In the case of plurality, $\alpha=(1,0,\ldots,0)$.  In the case of veto, $\alpha=(1,1,\ldots,1,0)$.
It is clear that plurality and veto are special cases of the scoring rule of $r$-approval.

\smallskip
\noindent\textbf{Weights of voters.}
Voters can have different weights. 
Let $w_j\in\mathbb{N}$ be the weight of voter $v_j$. In a weighted election, the total score of a candidate is the {\em weighted sum} of the scores a candidate receives from the voters. For example, candidate $c_i$ receives a score $w_j\cdot \alpha_{\tau_j^{-1}(i)}$ from voter $v_j$.

By re-indexing all of the candidates, we can set, without loss of generality, $c_m$ as the winner in the absence of bribery.

\smallskip
\noindent\textbf{Adversarial models.}
We consider an attacker that does not belong to $\mathcal{C}\cup \mathcal{V}$ but attempts to manipulate the election by bribing some voters.
Suppose voter $v_j$ has a {\em bribing price} $p_j^b$, meaning that $v_j$, upon receiving a bribery of amount $p_j^b$ from the attacker, will change its preference list to the list given by the attacker.
The attacker has a total budget $B$. As in the bribery problem, we also consider two kinds of attackers:
\begin{compactitem}
	\item {\em Constructive attacker}: This attacker attempts to make a designated candidate win the election, meaning that the designated candidate is the only candidate who gets the highest score. 
	\item {\em Destructive attacker}: This attacker attempts to make a designated candidate lose the election, meaning that there is another candidate that gets a strictly higher score than the designated candidate does.
\end{compactitem}

\smallskip
\noindent\textbf{Protection.} In the protection problem,
voter $v_j$, upon receiving an award of amount $p_j^a$ (or {\em awarding price}) from the defender, will always report its preference list faithfully and cannot be bribed. Note that $p_j^a$ may have multiple interpretations, such as monetary award, economic incentives or the cost of isolating voters from bribery.
We say a voter $v_j$ is {\em awarded} if $v_j$ receives an award of $p_j^a$.

\smallskip
\noindent\textbf{Problem statement.} We formalize our problem as follows.
\vspace{-3mm}
\begin{center}
	\fbox{\begin{minipage}{0.98\textwidth}
			\noindent\textbf{The constructive protection problem (i.e., protecting elections against constructive attackers):}
			
			\noindent{\em Input:} A set $\mathcal{C}$ of $m$ candidates. A set $\mathcal{V}$ of $n$ voters, each with a weight $w_j\in\mathbb{Z}_{>0}$, a preference list $\tau_j$, an awarding price of $p_j^a\in\mathbb{Z}_{>0}$ and a bribing price of $p_j^b\in\mathbb{Z}_{>0}$. A
			scoring rule for selecting a single winner. A defender with a defense budget $F\in\mathbb{Z}_{\ge 0}$. An attacker with an attack budget $B\in\mathbb{Z}_{\ge 0}$ attempting to make candidate $c_m$ win the election.
			
			\noindent{\em Output:} Decide whether there exists a $\VV_F\subseteq \VV$ such that
			\begin{itemize}
				\item  $\sum_{j:v_j\in\VV_F}p_j^a\le F$; and
				\item  for {\em any} subset $\VV_B\subseteq \VV\setminus\VV_F$ with $\sum_{j:v_j\in \VV_B}p_j^b\le B$, $c_m$ does not get a strictly higher score than any other candidate
				despite the attacker bribing the voters belonging to $\VV_B$ (i.e., bribing $\VV_B$).
			\end{itemize}
	\end{minipage}}
\end{center}
\vspace{-2mm}
\begin{center}
	\fbox{\begin{minipage}{0.98\textwidth}
			\noindent\textbf{The destructive protection problem (i.e., protecting elections against destructive attackers):}
			
			\noindent{\em Input:} A set $\mathcal{C}$ of $m$ candidates. A set $\mathcal{V}$ of $n$ voters, each with a weight $w_j\in\mathbb{Z}_{>0}$, a preference list $\tau_j$, an awarding price of $p_j^a\in\mathbb{Z}_{>0}$ and a bribing price of $p_j^b\in\mathbb{Z}_{>0}$. A scoring rule
			for selecting a single winner. Suppose $c_m$ is the winner if no voter is bribed. A defender with a defense budget $F\in\mathbb{Z}_{\ge 0}$. An attacker with an attack budget $B\in\mathbb{Z}_{\ge 0}$ attempting to make $c_m$ lose the election by making $c\in \mathcal{C}\setminus\{c_m\}$ get a strictly higher score than $c_m$ does.
			
			\noindent{\em Output:} Decide if there exists a $\VV_F\subseteq \VV$ such that
			\begin{itemize}
				\item  $\sum_{j:v_j\in\VV_F}p_j^a\le F$; and
				\item  for {\em any} subset $\VV_B\subseteq \VV\setminus\VV_F$ such that $\sum_{j:v_j\in \VV_B}p_j^b\le B$, no candidate $c\in \mathcal{C}\setminus\{c_m\}$ can get a strictly higher score than $c_m$ does despite the attacker bribing $\VV_B$.
			\end{itemize}
	\end{minipage}}
\end{center}
\vspace{-1mm}

\noindent\textbf{Further terminology and notations.}
We denote by $W(c_i)$ the total score obtained by candidate $c_i$ in the absence of bribery (i.e., no voter is bribed).
If the defender can select $\VV_F$ such that no constructive or destructive attacker can succeed, we say the defender succeeds. 
We call our problem as the (constructive or destructive) {\em weighted-\$-protection} problem, where ``weighted'' indicates that the voters are weighted and ``\$'' indicates that arbitrary awarding and bribing prices are involved.
In addition to investigating the general {\em weighted-\$-protection} problem, we also investigate the following special cases of it:
\begin{compactitem}
	\item the {\em \$-protection} problem with $w_j=1$ for each $j$ (i.e., the voters are not weighted);
	\item the {\em weighted-protection} problem with $p_j^a=p_j^b=1$ for each $j$ (i.e., voters are associated with the unit awarding price and the unit bribing price);
	\item the {\em unit-protection} problem with $w_j=p_j^a=p_j^b=1$ for each $j$ (i.e., voters are not weighted, and are associated with the unit awarding price and the unit bribing price).
	\item the {\em symmetric protection} problem with $p_j^a=p_j^b$ for each $j$ (i.e., the awarding price and the bribing price are always the same), while noting that different voters may have different prices.
\end{compactitem}










\vspace{-3mm}
\section{The Case of Constantly Many Candidates}
\label{sec:constant-many-candidates}



\subsection{The Weighted-\$-Protection Problem}\label{subsec:weighted-dollar}
The goal of this subsection is to prove the following theorem.
\begin{theorem}\label{thm:sigma-hard-general-destructive}
	For any non-trivial scoring rule, both the {\em constructive and destructive weighted-\$-protection problem}, is $\Sigma_2^p$-complete.
\end{theorem}
The theorem follows from Lemma~\ref{lemma:general-sigma-complete} and Lemma~\ref{lemma:general-hardness} below, which shows the $\Sigma_2^p$ membership and $\Sigma_2^p$-hardness, respectively.

\begin{lemma}\label{lemma:general-sigma-complete}
	For any non-trivial scoring rule, both the constructive and destructive weighted-\$-protection problems are in $\Sigma_2^p$.	
\end{lemma}

\begin{lemma}\label{lemma:general-hardness}
	For any non-trivial scoring rule, both the constructive and destructive weighted-\$-protection problems are both $\Sigma_2^p$-hard even if there are only $m=2$ candidates.
\end{lemma}
\smallskip
\noindent{\it Proof sketch.}
The proof of Lemma~\ref{lemma:general-hardness} follows from De-Negre (DNeg) variant of bi-level knapsack problem, which is proved to be $\Sigma_2^p$-hard by Caprara et al.~\cite{caprara2014study}. 
    We give a brief explanation.
	In this De-Negre variant, there are an adversary and a packer. The adversary has a reserving budget $\bar{F}$ and the packer has a packing budget $\bar{B}$. There is a set of $n$ items, each having a price $\bar{p}_j^a$ to the adversary, a price $\bar{p}_j^b$ to the packer, and a weight $\bar{w}_j=\bar{p}_j^b$ to both the adversary and the packer. The adversary first reserves a subset of items whose total price is no more than $\bar{F}$. Then the packer solves the knapsack problem with respect to the remaining items that are not reserved; i.e., the packer will select a subset of the remaining items whose total price is no more than $\bar{B}$ but their total weight is maximized. The De-Negre variant asks if the adversary can reserve a proper subset of items such that the total weight of the unreserved items that are selected by the packer is no more than some parameter $W$. The De-Negre variant is similar to the weighted-\$-protection protection problem, because we can view the defender and attacker in the protection problem respectively as the adversary and packer in the bi-level knapsack problem.
	In the case of a single-winner election with $m=2$ candidates, the goal of the defender is to assure that the constructive attacker cannot make the loser get a strictly higher score than the winner by bribing. This is essentially the same as ensuring that the constructive attacker cannot bribe a subset of non-awarded voters whose total weight is higher than a certain threshold, which is the same as the bi-level knapsack problem. \qed

\subsection{{The Weighted-Protection Problem}}\label{subsec:weighted-unit}
This is a special case of the weighted-\$-protection problem when $p_j^a=p_j^b=1$.

The following theorem used by Faliszewski et al. \cite{faliszewski2009hard} was originally proved for another problem. In our context, $F=0$ and thus $\VV_F=\emptyset$, it is \NP-hard to decide if the constructive attacker can succeed or, equivalently, if the defender {\em cannot} succeed. Hence, it is \textrm{coNP}-hard to decide if the defender can succeed and Theorem~\ref{thm:uni-price-restrict-np} follows.
\begin{theorem}\label{thm:uni-price-restrict-np}
	\emph{(By Faliszewski et al. \cite{faliszewski2009hard})}
	If $m$ is a constant, the constructive weighted-protection problem is \textrm{coNP}-hard for any
	scoring rule that $\alpha_2,\alpha_3,\ldots,\alpha_m$ are not all equal (i.e., it does not  hold that $\alpha_2=\alpha_3=\ldots=\alpha_m$).
\end{theorem}
In contrast, the destructive version is easy. Using the fact that $m$, the number of candidates, is a constant, we can prove the following Theorem~\ref{thm:uni-price-plu-p} through suitable enumerations.
\begin{theorem}\label{thm:uni-price-plu-p}
	If $m$ is a constant, then the destructive weighted-protection problem is in \PP\, for any scoring rule.
\end{theorem}

\vspace{-2mm}
\subsection{The \$-Protection Problem}\label{subsec:unit-dollar}

This is the special case of the protection problem with $w_j=1$ for every $j$. The following two theorems illustrate the significant difference (in terms of complexity) between the general problem and its special case with symmetricity (i.e., $p_j^a=p_j^b$). 


\begin{theorem}\label{thm:assy-uniweight-np}
	For constant $m$ and any non-trivial scoring rule, both the constructive and destructive \$-protection problems are \NP-complete.
\end{theorem}

\begin{theorem}
	\label{thm:uni-weight-plu-p}
	For constant $m$, both destructive and
	constructive symmetric \$-protection problems are in $P$ for any scoring rule.
\end{theorem}

\section{The Case of Arbitrarily Many Candidates}
\label{sec:arbitrary-many-candidates}

\subsection{The Case of Constructive Attacker}


The following Theorem \ref{thm:sigma-2-price} shows $\Sigma_2^p$-hardness for the most special cases of the constructive weighted-\$-protection problem, namely $w_j=p_j^a=p_j^b=1$ (unit-protection). It thus implies readily the $\Sigma_2^p$-hardness for the more general constructive \$-protection and constructive weighted-protection.

\begin{theorem}
	\label{thm:sigma-2-price}
	For arbitrary $m$, the $r$-approval constructive unit-protection problem is $\Sigma_2^p$-complete. 
\end{theorem}

Membership in $\Sigma_2^p$ follows directly from Lemma~\ref{lemma:general-sigma-complete}. To prove Theorem~\ref{thm:sigma-2-price}, it suffices to show the following.

\begin{lemma}
	\label{lemma:sigma-2-price}
	For arbitrary $m$, the $r$-approval constructive unit-protection problem is $\Sigma_2^p$-hard even if $r=4$.
\end{lemma}

To prove Lemma~\ref{lemma:sigma-2-price}, we reduce from a variant of the $\exists\forall$ 3 dimensional matching problem (or $\exists\forall$3DM), which is called $\exists\forall$3DM$'$ and defined below. The classical $\exists\forall$ 3DM is $\Sigma_2^p$-hard~ proved by Mcloughlin \cite{mcloughlin1984complexity}.
By leveraging the proof by Mcloughlin ~\cite{mcloughlin1984complexity}, we can show the $\Sigma_2^p$-hardness of the $\exists\forall$3DM$'$ problem.
\begin{center}
	\fbox{\begin{minipage}{0.98\textwidth}
			\textbf{$\exists\forall$3DM$'$:}
			Given a parameter $t$, three disjoint sets of elements $W$, $X$, $Y$ of the same cardinality, and two disjoint subsets $M_1$ and $M_2$ of $W\times X\times Y $ such that
			$M_1$ contains each element of $W\cup X\cup Y$ at most once. Does there exist a subset $U_1\subseteq M_1$ such that $|U_1|=t$ and for any $U_2\subseteq M_2$, $U_1\cup U_2$ is {\em not} a {\em perfect matching} (where a {\em perfect matching} is a subset of triples in which every element of $W\cup X\cup Y$ appears exactly once)?
	\end{minipage}}
\end{center}
\vspace{-2mm}
\begin{proof}[Proof of Lemma~\ref{lemma:sigma-2-price}]
	Given an arbitrary instance of $\exists\forall$ 3DM$'$, we construct an instance of the constructive unit-protection problem in $r$-approval election as follows. Recall that $r=4$ and thus every voter votes for 4 candidates.
	
	Suppose $|W|=|X|=|Y|=n$, $|M_1|=m_1$, $|M_2|=m_2$.	
	
	There are $3n+2$ key candidates, including:
	\begin{compactitem}
		\item $3n$ key candidates, each corresponding to one distinct element of $W\cup X\cup Y$ and we call them element candidates. The score of every element candidate is $n+\xi$;
		\item one key candidate called leading candidate, whose total score is $n+t+\xi-1$;
		\item one key candidate called designated candidate, whose total score is $\xi$.
	\end{compactitem}
	Here $\xi$ is some sufficiently large integer, e.g., we can choose $\xi=(m_1+m_2)n$.
	Besides key candidates, there are also 
	many dummy candidates, each of score either $1$ or $m_1-t+1$. The number of dummy candidates will be determined later.
	
	
	There are $m_1+m_2(m_1-t+1)$ key voters, including: 
	\begin{compactitem}
		\item $m_1$ key voters, each corresponding to a distinct triple in $M_1$ and we call them $M_1$-voters. For each $(w_i,x_j,y_k)\in M_1$, the corresponding voter votes for the $3$ candidates corresponding to elements $w_i$, $x_j$, $y_k$ together with the leading candidate;
		\item $m_2\cdot(m_1-t+1)$ key voters, each distinct triple in $M_2$ corresponds to exactly $m_1-t+1$ voters and we call them $M_2$-voters. For every $(w_i,x_j,y_k)\in M_2$, each of its $m_1-t+1$ corresponding voters vote for the $3$ candidates corresponding to elements $w_i$, $x_j$, $y_k$ together with one distinct dummy candidate. Since the $m_1-t+1$ voters are identical,
		we can view them as $m_1-t+1$ copies, i.e., every $M_2$-voter has $m_1-t+1$ copies.
	\end{compactitem}
	Besides key voters, there are also sufficiently many dummy voters. Each dummy voter votes for exactly one key candidate and 3 distinct dummy candidates. Dummy voters and dummy candidates are used to make sure that the score of key candidates are exactly as we have described. More precisely, if we only count the scores of key candidates contributed by key voters, then the element candidate corresponding to $z\in W\cup X\cup Y$ has a score of $d(z)=d_1(z)+(m_1-t+1)d_2(z)$ where $d_i(z)$ is the number of occurrences of $z$ in the triple set $M_i$ for $i=1,2$, and the leading candidate has a score of $m_1$. Hence, there are exactly $n+\xi-d(z)$ dummy voters who vote for the element candidate corresponding to $z$, and $n+t+\xi-1-m_1$ dummy voters who vote for the leading candidate.
	
	Overall, we create $\sum_{z\in W\cup X\cup Y} (n+\xi-d(z))+n+t+\xi-m_1-1$ dummy voters, and $3\sum_{z\in W\cup X\cup Y} (n+\xi-d(z))+3(n+t+\xi-m_1-1)+m_2$ dummy candidates. 

As the leading candidate is the current winner, the constructive unit-protection problem asks whether the election can be protected against an attacker attempting to make the designated candidate win. The defense budget is $F=m_1-t$ and the attack budget is $B=n$. In the following we show that the defender succeeds if and only if the given $\exists\forall$ 3DM$'$ instance admits a feasible solution $U_1$.
	
	\smallskip
	\noindent\textbf{``Yes" Instance of $\exists\forall$ 3DM$'$ $\to$ ``Yes" Instance of Constructive Unit-Protection.} Suppose the instance of $\exists\forall$ 3DM$'$ admits a feasible solution $U_1$, we show that the answer for constructive unit-protection problem is ``Yes".
	
	Recall that each $M_1$-voter corresponds to a distinct triple $(w_i,x_j,y_k)$ in $M_1$ and votes for 4 candidates -- the leading candidate and the three candidates corresponding to $w_i,x_j,y_k$. We do {\em not} award $M_1$-voters {corresponding} to the triples in $U_1$, but award all of the remaining $M_1$-voters. {The resulting} cost is exactly $F=m_1-t$. In what follows we show that after awarding voters this way, the attacker cannot make the designated candidate win.
	
	Suppose on the contrary, the attacker can make the designated candidate win by bribing $\alpha\le t$ voters among the $M_1$-voters, $\beta\le m_2$ voters among the $M_2$-voters, and $\gamma$ dummy voters. We claim that the following inequalities hold:
	\begin{subequations}
		\begin{align}
		&\alpha+\beta+\gamma\le n~\label{3dm'-1}\\
		&4\alpha+3\beta+\gamma\ge 3n+t ~\label{3dm'-2}
		\end{align}
	\end{subequations}
	Inequality~(\ref{3dm'-1}) follows from the fact that the attack budget is $n$ and the attacker can bribe at most $n$ voters. Inequality~(\ref{3dm'-2}) holds because of the following. Given that a candidate can get at most one score from each voter and that the attacker can bribe at most $n$ voters, bribing voters can make the designated candidate obtain a score at most $n+\xi$. Hence, the score of each {key candidate} other than the designated one should be at most $n+\xi-1$. Recall that without bribery, each of the $3n$ {element candidate} has a score of $n+\xi$ and the leading candidate has a score of $n+t+\xi-1$. Hence, the attacker should decrease at least 1 score from each element candidate and $t$ scores from the leading candidate, leading to a total score of $3n+t$. Note that an {$M_1$-voter} contributes 1 score to 4 key candidates, therefore it contributes in total a score of $4$ to the key candidates. Similarly an {$M_2$-voter} contributes a score of $3$, and a dummy voter contributes a score of $1$ to the key candidates. Therefore, by bribing (for example) an $M_1$-voter, the total score of all the element candidates and the leading candidate can decrease by at most $4$. Thus, inequality~(\ref{3dm'-2}) holds.

	In the following we derive a contradiction based on Inequalities (\ref{3dm'-1}) and (\ref{3dm'-2}).
	By plugging $\gamma\le n-\alpha-\beta$ into Inequality (\ref{3dm'-2}), we have
	$3\alpha+2\beta\ge 2n+t.$
	Since $\beta\le n-\alpha$, we have
	$3\alpha+2\beta\le \alpha+2n\le 2n+t$.
	Hence, $3\alpha+2\beta= \alpha+2n= 2n+t$, and we have $\alpha=t$ and $\beta=n-t$. Note that the defender has awarded every {$M_1$-voter} except the ones corresponding to $U_1$, where $|U_1|=t$. Hence, every voter corresponding to the triples in $U_1$ is bribed. Furthermore, as Inequality (\ref{3dm'-2}) is tight, bribing voters makes the designated candidate have a score of $n+\xi$, while making each of the other key candidates have a score of $n+\xi-1$. This means that the score of each element candidate decreases exactly by $1$. Hence, the attacker has selected a subset of $M_2$-voters such that together with the $M_1$-voters corresponding to triples in $U_1$, these voters contribute exactly a score of 1 to every element candidate. Let $U_2$ be the set of triples to which the bribed $M_2$-voters correspond, then $U_1\cup U_2$ forms a 3-dimensional matching, which is a contradiction to the fact that $U_1$ is a feasible solution to the $\exists\forall$ 3DM$'$ instance. Thus, the attacker cannot make the designated candidate win and the answer for the constructive unit-protection problem is ``Yes".
	
	\smallskip
	\noindent\textbf{\say{No} Instance of $\exists\forall$ 3DM$'$ $\to$ \say{No} Instance of Constructive Unit-Protection.} Suppose for any $U_1\subseteq M_1$, $|U_1|=t$ there exists $U_2\in M_2$ such that $U_1\cup U_2$ is a perfect matching, we show that the answer to the constructive unit-protection problem is ``No". Consider an arbitrary set of voters awarded by the defender. Among the awarded voters, let $H$ be the set of triples that corresponds to the awarded $M_1$-voters. As $|H|\le m_1-t$, $|M_1\setminus H|\ge t$. We select an arbitrary subset $H_1\subseteq M_1\setminus H$ such that $|H_1|=t$. There exists some $H_2\subseteq M_2$ such that $H_1\cup H_2$ is a perfect matching, and we let the attacker bribe the set of voters corresponding to triples in $H_1\cup H_2$. Note that this is always possible as every $M_2$-voter has $m_1-t+1$ copies, so no matter which $M_2$-voters are awarded the briber can always select one $M_2$-voter corresponding to each triple in $H_2$. It is easy to see that by bribing these voters, the score of every element candidate decreases by $1$, and the score of the leading voter decreases by $t$. Meanwhile, let each bribed voter vote for the designated candidate and three distinct dummy candidates, then the designated candidate has a score of $n+\xi$ and becomes a winner, i.e., the answer to the constructive unit-protection problem is ``No".	
	\qed
\end{proof}
\hspace{-1mm}
\noindent\textbf{Remark}. The proof of Lemma~\ref{lemma:sigma-2-price} can be easily modified to prove the $\Sigma_2^p$-hardness of $r$-approval constructive unit-protection problem for any fixed $r\ge 4$. Specifically, we can make the same reduction, and add dummy candidates such that every voter additionally votes for exactly $r-4$ distinct dummy candidates. 
\vspace{-2mm}
\subsection{The Case of Destructive Attacker}
\label{subsec: destructive-arbitrary}

\begin{theorem}
	\label{thm:destructive-np}
	Both $r$-approval destructive weighted-protection and $r$-approval (symmetric) \$-protection problems are \NP-complete.
\end{theorem}
The proof of Theorem~\ref{thm:destructive-np} is based on a crucial observation of the equivalence between the destructive weighted-\$-protection problem (under an arbitrary scoring rule) and the {\it minmax vector addition problem} we introduce (see Appendix~\ref{appsec:dominance}). The full proof of Theorem~\ref{thm:destructive-np} can be found in Appendix~\ref{appsec:des}.


\vspace{-2mm}
\section{Summary of Results}
\label{sec:discussion}
\vspace{-2mm}
The preceding characterization of the computational complexity of the protection problem in various settings is summarized in Table~\ref{table:1}.
	\begin{table*}[!htbp]
\caption{Summary of results for {\em single-winner} election under the {\em $r$-approval} scoring rule: 
``Symmetric'' means $p_j^a=p_j^b$ for every $j$ and ``asymmetric'' means otherwise;
hardness results that are proved for the case with only two candidates (i.e., $m=2$) are marked with a ``$\diamond$'' (Note that when $m=2$, the $1$-approval rule is the same as the plurality, veto or Borda scoring rule. It can be shown that with a slight modification, the hardness results hold for any {\em non-trivial} scoring rule);
algorithmic results (marked with a ``\PP'') hold for arbitrary scoring rules;
the complexity of the protection problem against a destructive attacker with $w_j=p_j^a=p_j^b=1$ remains open;
for most variants of the protection problem against a constructive attacker, we only provide hardness results and we do not know yet whether or not they belong to the class of \text{coNP}-complete or $\Sigma_2^p$-complete proble.\label{table:1}
}
\centering
\label{my-label}
\scalebox{0.85}{
\begin{tabular}{|c|c|c|c|}
\hline
\# of candidates & Model parameters & Destructive attacker & Constructive attacker \\ \hline
\multirow{5}{*}{constant}  & Weighted, Priced, Asymmetric  & $\Sigma_2^p$-complete $\hspace{1mm}\diamond$ (Thm~\ref{thm:sigma-hard-general-destructive}) & $\Sigma_2^p$-complete $\hspace{1mm}\diamond$ (Thm~\ref{thm:sigma-hard-general-destructive}) \\ \cline{2-4}
& Weighted, $p_j^a=p_j^b=1$  & \PP\, (Thm~\ref{thm:uni-price-plu-p}) & \textrm{coNP}-hard (Thm~\ref{thm:uni-price-restrict-np})  \\ \cline{2-4}
&  $w_j=1$, Priced, Asymmetric & $\text{\NP}$-complete $\hspace{1mm}\diamond$ (Thm~\ref{thm:assy-uniweight-np}) & $\text{\NP}$-complete $\hspace{1mm}\diamond$ (Thm~\ref{thm:assy-uniweight-np}) \\ \cline{2-4}
& $w_j=1$, Priced, Symmetric & \PP\, (Thm~\ref{thm:uni-weight-plu-p}) & \PP\, (Thm~\ref{thm:uni-weight-plu-p}) \\ \cline{2-4}
&  $w_j=1$, $p_j^a=p_j^b=1$ & \PP\,(Thm~\ref{thm:uni-weight-plu-p}) & \PP\,(Thm~\ref{thm:uni-weight-plu-p}) \\ \hline
\multirow{5}{*}{arbitrary}  & Weighted, Priced, Asymmetric  & $\Sigma_2^p$-complete $\hspace{1mm}\diamond$ (Thm~\ref{thm:sigma-hard-general-destructive}) & $\Sigma_2^p$-complete $\hspace{1mm}\diamond$ (Thm~\ref{thm:sigma-hard-general-destructive}) \\ \cline{2-4}
& Weighted, $p_j^a=p_j^b=1$  & \NP-complete (Thm~\ref{thm:destructive-np}) & $\Sigma_2^p$-hard (Thm~\ref{thm:sigma-2-price})\\ \cline{2-4}
&  $w_j=1$, Priced, Asymmetric & $\text{\NP}$-complete  (Thm~\ref{thm:destructive-np})  & $\Sigma_2^p$-hard (Thm~\ref{thm:sigma-2-price}) \\ \cline{2-4}
&  $w_j=1$, Priced, Symmetric & $\text{\NP}$-complete  (Thm~\ref{thm:destructive-np}) & $\Sigma_2^p$-hard (Thm~\ref{thm:sigma-2-price})\\ \cline{2-4}
&  $w_j=1$, $p_j^a=p_j^b=1$ & ? & $\Sigma_2^p$-hard (Thm~\ref{thm:sigma-2-price}) \\ \hline
 \end{tabular}}
\end{table*}

		We remark three natural open problems for future research. One is the complexity of the destructive protection problem with $w_j=p_j^a=p_j^b=1$. It is not clear whether the problem is in \PP\, or is \NP-complete. Another is the constructive protection problem with $p_j^a=p_j^b=1$ and arbitrary voter weights. We only show its coNP-hardness, it is not clear whether or not this problem is coNP-complete. The third problem is the complexity of $r$-approval constructive unit-protection problem when $r=\{1,2,3\}$ as our hardness proof only holds when $r\ge 4$.

\vspace{-3mm}
\section{Conclusion}
\label{sec:conclusion}
\vspace{-2mm}
We introduced the protection problem and characterized its computational complexity.
We showed that the problem, in general, is $\Sigma_2^p$-complete, and identified settings in which the problem becomes easier.
Moreover, we showed the protection problem in some parameter settings is polynomial-time solvable, suggesting that these parameter settings can be used for real-work election applications.

In addition to the open problems mentioned in Section~\ref{sec:discussion}, the following are also worth investigating.
First, our hardness results would motivate the study of approximation or FPT (fixed parameter tractable) algorithms for the protection problem. Note that even polynomial time approximation schemes can exist for $\Sigma_2^p$-hard problems (see, e.g., By Caparara et al. \cite{caprara2014study}). It is thus desirable that a similar result can be obtained for some variants of the protection problem.
Second, how effective is this approach when applied towards the problem of defending against other types of attackers that can, e.g., add or delete votes?
Third, much research remains to be done in extending the protection problem to accommodate other scoring rules such as Borda and Copeland.


\clearpage

\normalsize

\section{Appendix}
\label{sec:Apeendix}

\subsection{Destructive Weighted-Protection -- an Equivalent Formulation}\label{appsec:dominance}
We provide an equivalent formulation of the destructive weighted-protection problem under any scoring rule $\alpha=(\alpha_1,\cdots,\alpha_m)$, which will be very useful for several proofs throughout this paper. 

\begin{center}
	\fbox{\begin{minipage}{0.97\textwidth}
			\noindent\textbf{The minmax vector addition problem:} 
			
			\noindent{\em Input:}
			A vector $\vec{\Lambda}=(\Lambda(c_1),\Lambda(c_2), \ldots,\Lambda(c_{m-1}))$ where $\Lambda(c_i)$ is the score of $c_i$ in the absence of bribery. An $(m-1)$-vector $\vec{\Delta}_j=(\Delta_{1j},\Delta_{2j},\ldots,\Delta_{(m-1),j})$ for each voter $v_j$ where $\Delta_{ij}= \alpha_1-\alpha_{\tau_j^{-1}(i)}+\alpha_{\tau_j^{-1}(m)}-\alpha_{m}$. Awarding price $p_j^a$ and bribing price $p_j^b$ for voter $v_j$, $j=1,2,\ldots,n$. 
			Defense budget $F$ and attack budget $B$. 
			
			{\em Output}: Decide if there exists a subset $\VV_F\subseteq \VV$ such that
			\begin{itemize}
				\item $\sum_{j:v_j\in \in\VV_F} p_j^a\le F$; and
				\item For {\em any} subset $\VV_B\subseteq \VV\setminus\VV_F$ with $\sum_{j:v_j\in \VV_B} p_j^b\le B$, it holds that
				$$\left\|\vec{\Lambda}+\sum_{j:v_j\in\VV_B} w_j \vec{\Delta}_j\right\|_{\infty}\le \Lambda(c_m),$$
				where $\left\|\cdot\right\|_{\infty}$ is the infinity norm (i.e., the maximal absolute value among the $m-1$ coordinates).
			\end{itemize}
	\end{minipage}}
\end{center}

\begin{lemma}\label{lemma:equivalent}
	The answer to the destructive weighted-\$-protection problem is ``Yes" if and only if the answer to the corresponding minmax vector addition problem is ``Yes''.
	\end{lemma}
\begin{proof}
	\noindent\textbf{A ``Yes" Instance of Minmax Vector Addition $\rightarrow$ A ``Yes" Instance of Destructive Weighted-\$-Protection.} 
	Suppose the answer to the minmax vector addition problem is ``Yes.'' Then, there exists some $\VV_F^*\subseteq \VV$ such that for any $\VV_B\subseteq\VV\setminus\VV_F^*$ with $\sum_{j:v_j\in\VV_B}p_j^b\le B$, it holds that 
	\begin{eqnarray}\label{eq:1}
	\left\|\vec{\Lambda}+\sum_{j:v_j\in\VV_B} w_j\vec{\Delta}_j\right\|_{\infty}\le \Lambda(c_m).
	\end{eqnarray}
For showing a contradiction, suppose the answer for the destructive weighted-\$-protection problem is ``No''. In this case, even if the defender awards {the voters in} $\VV_F^*$, the attacker can still make $c_m$ lose by bribing some subset $\VV_B^*$ of voters. Note that if $c_m$ does not win, there must exist some other candidate, say, $c_i$, who gets a strictly higher score than $c_m$ after the attacker bribes some voters. Let us compare their scores before and after bribing voters. Before bribing voters, the scores of $c_i$ and $c_m$ are $\Lambda(c_i)$ and $\Lambda(c_m)$, respectively. Recall that a candidate $c_k$ is at the position of $\tau_j^{-1}(k)$ on the preference list of $v_j$, therefore any $v_j\in \VV_B^*$ contributes {a score of} $\alpha_{\tau_j^{-1}(i)}$ to $\Lambda(c_i)$, and contributes {a score of} $\alpha_{\tau_j^{-1}(m)}$ {to $\Lambda(c_m)$}. After bribing voters, the preference list of $v_j$ is changed, but regardless of the change, $v_j$ contributes at least $\alpha_m$ to $c_m$ and at most $\alpha_1$ to $c_i$. Let the scores of $c_i$ and $c_m$ after bribing voters be $\Lambda'(c_i)$ and $\Lambda'(c_m)$, respectively. Then, it follows that
   \begin{eqnarray*}
   && \Lambda'(c_i)\le \Lambda(c_i)+\sum_{j:v_j\in \VV_B^*} w_j (\alpha_1-\alpha_{\tau_j^{-1}(i)})\\
   && \Lambda'(c_m)\ge \Lambda(c_m)+\sum_{j:v_j\in \VV_B^*}w_j(\alpha_m-\alpha_{\tau_j^{-1}(m))}.
   \end{eqnarray*}
	Since $\Lambda'(c_i)>\Lambda'(c_m)$, we have 
	$$\Lambda(c_i)+\sum_{j:v_j\in \VV_B^*} w_j(\alpha_1-\alpha_{\tau_j^{-1}(i)})>\Lambda(c_m)+\sum_{j:v_j\in \VV_B^*}w_j(\alpha_m-\alpha_{\tau_j^{-1}(m)}),$$
that is, 
$$\Lambda(c_i)+\sum_{j:v_j\in \VV_B^*}w_j\Delta_{ij}>\Lambda(c_m),$$
which contradicts Eq.~(\ref{eq:1}). Thus, the answer to the destructive weighted-\$-protection problem is ``Yes".

\medskip

\noindent\textbf{A ``Yes" Instance of Destructive Weighted-\$-Protection $\rightarrow$ A ``Yes" Instance of Minmax Vector Addition.} Suppose the answer to the destructive weighted-\$-protection problem is ``Yes" by awarding the voters in $\VV_F^*$. We show that the answer to the corresponding instance of minmax vector addition problem is ``Yes''. Suppose on the contrary the answer is ``No.'' Then, for $\VV_F^*$ there exists some $\VV_B^*\subseteq\VV\setminus \VV_F^*$ such that
\begin{eqnarray*}
\left\|\vec{\Lambda}+\sum_{j:v_j\in\VV_B^*} w_j\vec{\Delta}_j\right\|_{\infty}> \Lambda(c_m).
\end{eqnarray*}
Consequently, there must exist some $1\le i\le m-1$ such that 
$\Lambda(c_i)+\sum_{j:v_j\in\VV_B^*} \Delta_{ij}>\Lambda(c_m)$. By plugging in $\Delta_{ij}$, we have 
$$\Lambda(c_i)+\sum_{j:v_j\in \VV_B^*} w_j(\alpha_1-\alpha_{\tau_j^{-1}(i)})>\Lambda(c_m)+\sum_{j:v_j\in \VV_B^*}w_j (\alpha_m-\alpha_{\tau_j^{-1}(m)}).$$
	This means that if the defender awards the voters in $\VV_F^*$, then the attacker can bribe the voters in $\VV_B^*$ to change their preference lists such that for any $v_j\in \VV_B^*$, candidate $c_i$ is on top of the list and $c_m$ is at bottom of the list. By doing this, $c_i$ gets a strictly higher score than $c_m$. This contradicts the fact that the answer to the destructive weighted-\$-protection problem is ``Yes". Hence, the answer to the minmax vector addition problem is ``Yes''.
	\qed
		\end{proof}

\subsection{Voter Dominance and Preliminary Observations}\label{appsec:pre}

Let $S_m$ be the set of permutations over $\{c_1,c_2,\ldots,c_m\}$.
Each element of $S_m$ can be a preference list. Let
$\VV^h\subseteq\VV$ be the set of voters whose preference list is the $h$-th element of $S_m$.

For two voters $v_j$ and $v_{j'}$ , we say $v_j$ {\em dominates} $v_{j'}$ (or $v_{j'}$ is {\em dominated} by $v_j$), denoted by $v_{j'}\prec v_{j}$, {if any of the following two conditions hold}:
(i) The following holds and at least one of the inequalities is strict:
\begin{equation*}
(\tau_j=\tau_{j'})~\wedge~(w_j\ge w_{j'}) ~\wedge~ (p_j^a\le p_{j'}^a) ~\wedge~ (p_{j}^b\le p_{j'}^b).
\end{equation*}
(ii) The following holds:
\begin{eqnarray*}
(\tau_j=\tau_{j'})\wedge(w_j=w_{j'}) \wedge (p_j^a=p_{j'}^a) \wedge (p_j^b=p_{j'}^b) \wedge (j'<j).
\end{eqnarray*}
{Note that the domination relation is only defined between the voters who have the {\em same} preference.} Intuitively, if $v_{j'}\prec v_{j}$, then $v_j$ is more ``important'' than $v_{j'}$ because $v_j$ has a greater weight but is ``cheaper'' to bribe or award (i.e., more valuable to both the attacker and the defender).


We have the following lemmas.

\begin{lemma}\label{lemma:destructive-1}
Consider the destructive weighted-\$-protection problem with $\VV_F\subseteq \VV$ being the set of awarded voters.
Suppose the attacker can succeed by bribing a subset $\VV_B\subseteq \VV\setminus \VV_F$ of voters. If $v_{j'}\prec v_{j}$, $v_{j'}\in \VV_B$ and $v_j\not\in(\VV_F\cup \VV_B)$, then the attacker can succeed by bribing $(\VV_B\setminus\{v_{j'}\})\cup\{v_j\}$.
\end{lemma}

Towards the proof, we need Lemma~\ref{lemma:equivalent}, which states that the destructive weighted-\$-protection problem is equivalent to another problem called minmax vector addition. The reader may refer to Section~\ref{subsec: destructive-arbitrary} for the definition of this problem. The following observation follows directly from the definition of $\Delta_{ij}$ which is included in the definition of minmax vector addition.
\begin{obse}
	If $v_{j'}\prec v_{j}$, then
	$\Delta_{ij'}\le \Delta_{ij}$ for $1\le i\le m-1$.
\end{obse}

Assuming Lemma~\ref{lemma:equivalent}, we can prove Lemma~\ref{lemma:destructive-1} as follows.

\begin{proof}[Proof of Lemma~\ref{lemma:destructive-1}]
We prove the lemma by applying an exchange argument {to} the minmax vector addition problem, which is equivalent to the destructive weighted-\$-protection by Lemma~\ref{lemma:equivalent}. Suppose by bribing voters in $\VV_B$ the destructive attacker can make $c_m$ lose. Then, it follows that $\sum_{j:v_j\in \VV_B} p_j^b\le B$ and	$$||\vec{\Lambda}+\sum_{j:j\in\VV_B} w_j\vec{\Delta}_j||_{\infty}> \Lambda(c_m).$$
	As $v_j$ dominates $v_{j'}$, $\Delta_{ij}\ge \Delta_{ij'}$ and $w_j\le w_{j'}$. Hence,  $$\sum_{j:j\in(\VV_B\setminus\{v_{j'}\})\cup\{v_j\}} p_j^b\le B$$	and
	$$||\vec{\Lambda}+\sum_{j:j\in(\VV_B\setminus\{v_{j'}\})\cup\{v_j\}} w_j\vec{\Delta}_j||_{\infty}> \Lambda(c_m).$$
	That is, the briber can also win by bribing voters in $(\VV_B\setminus\{v_{j'}\})\cup\{v_j\}$. 
	\qed
\end{proof}

\begin{lemma}\label{lemma:destructive-2}
Consider the destructive weighted-\$-protection problem. Suppose the defender succeeds by awarding a subset $\VV_F\subseteq \VV$ of voters. If $v_{j'}\prec v_{j}$, $v_{j'}\in \VV_F$ and $v_j\not\in\VV_F$, then the defender can succeed by awarding $(\VV_F\setminus\{v_{j'}\})\cup\{v_j\}$.
\end{lemma}
\begin{proof}
We again use an exchange argument {to} the minmax vector addition problem. Let $\VV_A=\VV_F\setminus\{v_{j'}\}$. Suppose on the contrary that the defender cannot win by fixing voters in $\VV_A\cup\{v_j\}$, then there exists some $\VV_B\subseteq \VV\setminus (\VV_A\cup\{v_j\})$ such that  $\sum_{j:v_j\in \VV_B} p_j^b\le B$ and
	$$||\vec{\Lambda}+\sum_{j:j\in\VV_B} w_j\vec{\Delta}_j||_{\infty}> \Lambda(c_m).$$
	We argue that the defender cannot win either by fixing voters in $\VV_A\cup\{v_{j'}\}$, which is a contradiction. Suppose the defender fixes voters in $\VV_A\cup\{v_{j'}\}$. There are two possibilities. If $v_{j'}\not\in\VV_B$, then we let the briber bribe voters in $\VV_B$. It is obvious that the briber can win. Otherwise, $v_{j'}\in \VV_B$, then we let the briber bribe voters in $(\VV_B\setminus\{v_{j'}\})\cup \{v_j\}$. Since $v_j$ dominates $v_{j'}$, we have
	$\sum_{j:v_j\in (\VV_B\setminus\{v_{j'}\})\cup \{v_j\}} p_j^b\le B$ and
	$$||\vec{\Lambda}+\sum_{j:j\in(\VV_B\setminus\{v_{j'}\})\cup \{v_j\}} w_j\vec{\Delta}_j||_{\infty}> \Lambda(c_m).$$
	Hence, the lemma is true.
	\qed
\end{proof}

We say $\bar{\VV}\subseteq \VV$ is maximal (with respect to $\VV$) if for any $\bar{v}\in \bar{\VV}$, there is no $v\in \VV\setminus \bar{\VV}$ that can dominate $\bar{v}$. That is, $\bar{\VV}$ contains the most important voters. The following corollary follows directly from the preceding two lemmas.

\begin{coro}\label{coro:maximal}
Consider the destructive weighted-\$-protection problem.
Without loss of generality, we can assume that $\VV_F$ is maximal with respect to $\VV$ and $\VV_B$ is maximal with respect to $\VV\setminus\VV_F$.
\end{coro}

Unfortunately, Corollary~\ref{coro:maximal} does not hold for the constructive weighted-\$-protection problem, which is significantly different from the destructive version of the protection problem in terms of computational complexity. Nevertheless, similar results hold for the unweighted constructive problem. 

\begin{lemma}\label{lemma:constructive-bribe-dominate}
	Given $\VV_F\subseteq \VV$ as the set of fixed voters in the constructive \$-bribery-protection problem, suppose a briber can make $c_i$ win by bribing a subset $\VV_B\subseteq \VV\setminus \VV_F$ of voters. If $v_{j'}\prec v_{j}$, $v_{j'}\in \VV_B$ and $v_j\not\in(\VV_F\cup \VV_B)$, then the briber can also win by bribing voters in $(\VV_B\setminus\{v_{j'}\})\cup\{v_j\}$.
\end{lemma}
\begin{proof}
	Again, we prove by an exchange argument. Suppose by bribing voters in $\VV_B$ the constructive briber can make $c_i$ win. Now we consider the following procedure: we change the preference list of $v_j$ into the same one as that of $v_{j'}$, and meanwhile, restore the preference list of $v_{j'}$ to the original one. As voters have the same weight, this procedure does not change the total score of every candidate, and $c_i$ is thus still the winner. Furthermore, this procedure is equivalent as we bribe $(\VV_B\setminus\{v_{j'}\})\cup\{v_j\}$. Since $v_j$ dominates $v_{j'}$, the total cost of bribing $(\VV_B\setminus\{v_{j'}\})\cup\{v_j\}$ is no more than that of bribing $\VV_B$. Hence, the lemma is true.
	\qed
\end{proof}

\begin{lemma}
	In the constructive \$-bribery-protection problem,	suppose the defender can win by fixing a subset $\VV_F\subseteq \VV$ of voters. If $v_{j'}\prec v_{j}$, $v_{j'}\in \VV_F$ and $v_j\not\in\VV_F$, then the defender can also win by fixing voters in $(\VV_F\setminus\{v_{j'}\})\cup\{v_j\}$.
\end{lemma}
\begin{proof}
	Suppose on the contrary that the defender cannot win by fixing voters in $\VV_F'=(\VV_F\setminus\{v_{j'}\})\cup\{v_j\}$. Then the constructive briber can win by bribing voters in some subset $\VV_B\subseteq \VV\setminus\VV_F'$. There are two possibilities. If $v_{j'}\not\in\VV_B$, then even if the defender fixes $\VV_F$ the briber can still bribe $\VV_B$ and make $c_i$ win, which is a contradiction. Otherwise $v_{j'}\in\VV_B$. If the defender fixes $\VV_F$, we let the briber bribe $(\VV_B\setminus\{v_{j'}\})\cup\{v_j\}$. According to Lemma~\ref{lemma:constructive-bribe-dominate}, if the briber can win by bribing $\VV_B$, he/she can also win by bribing $(\VV_B\setminus\{v_{j'}\})\cup\{v_j\}$, again contradicting the fact that the defender can succeed by awarding $\VV_F$.
	\qed
\end{proof}

The above Lemmas implies the following.

\begin{coro}\label{coro:maximal2}
	Without loss of generality, we can assume that $\VV_F$ is maximal with respect to $\VV$, and $\VV_B$ is maximal with respect to $\VV\setminus\VV_F$ in the constructive \$-bribery-protection problem.
\end{coro}

\subsection{Proofs Omitted in Section~\ref{subsec:weighted-dollar}}
Recall that our goal is to prove Theorem~\ref{thm:sigma-hard-general-destructive}.

\begingroup
\def\thetheorem{\ref{thm:sigma-hard-general-destructive}}
\begin{theorem}
	For any non-trivial scoring rule, both the {\em constructive and destructive weighted-\$-protection problem}, is $\Sigma_2^p$-complete.
\end{theorem}
\addtocounter{theorem}{-1}
\endgroup

We first show $\Sigma_2^p$-membership.


\begingroup
\def\thelemma{\ref{lemma:general-sigma-complete}}
\begin{lemma}
	For any non-trivial scoring rule, both the constructive and destructive weighted-\$-protection problems are in $\Sigma_2^p$.
\end{lemma}
\addtocounter{lemma}{-1}
\endgroup

For ease of proof, we use the following definition of $\Sigma_2^p$ from~\cite{wrathall1976complete} (see Theorem 3 therein).

\begin{defi}
\emph{(By Wrathall \cite{wrathall1976complete})}
Let $\Gamma$ be a finite set of symbols (alphabet) and $\Gamma^+$ be the set of strings of symbols in $\Gamma$. Let $L\subseteq \Gamma^+$ be a language. $L\in \Sigma_2^p$ if and only if there exists polynomials $\phi_1$, $\phi_2$ and a language $L'\in P=\Sigma_0^p$ such that for all $x\in \Gamma^+$,
$$x\in L ~~~~\text{ if and only if }~~~~ (\exists y_1)_{\phi_1} (\forall y_2)_{\phi_2} [(x,y_1,y_2)\in L'],$$
where $(\exists y_1)_{\phi_1} (\forall y_2)_{\phi_2} [(x,y_1,y_2)\in L']$ denotes
$$(\exists y_1)(\forall y_2) [|y_1|\le \phi_1(|x|) \text{ and if } |y_2|\le \phi_2(|x|), (x,y_1,y_2)\in L']. $$
\end{defi}

{
	\begin{proof}[Proof of Lemma~\ref{lemma:general-sigma-complete}]
		Given an instance $I$ of the constructive or destructive weighted-\$-protection problem, we want to know if there exists a subset $\VV_F$ such that $\sum_{j:v_j\in \VV_F} p_j^a\le F$ and for any subset $\VV_B\subseteq \VV\setminus\VV_F$ with $\sum_{j:v_j\in \VV_B}p_j^b\le B$ and any preference list $\hat{\tau}_j$ for $v_j\in \VV_B$, the following property $\mathcal{R}(I,\VV_F,\VV_B\cup \{\hat{\tau}_j|v_j\in \VV_B\})$ is true: By bribing voter in $\VV_B$ and change the preference list of each $v_j\in \VV_B$ to $\hat{\tau}_j$, a constructive attacker cannot make candidate $c_i$ win, or a destructive attacker cannot make $c_m$ lose. It is easy to see that the property $\mathcal{R}(I,\VV_F,\VV_B\cup \{\hat{\tau}_j|v_j\in \VV_B\})$ can be verified in polynomial time, therefore Lemma~\ref{lemma:general-sigma-complete} is proved.
		\qed
	\end{proof}
}
Now we prove the $\Sigma_2^p$-hardness.
\begingroup
\def\thelemma{\ref{lemma:general-hardness}}
\begin{lemma}
		For any non-trivial scoring rule, both the constructive and destructive weighted-\$-protection problems are both $\Sigma_2^p$-hard even if there are only $m=2$ candidates.
\end{lemma}
\addtocounter{lemma}{-1}
\endgroup
Note that in case of $m=2$, destructive weighted-\$-bribery-protection is equivalent to constructive weighted-\$-bribery-protection. We prove the Lemma~\ref{lemma:general-hardness} for the protection problem under plurality. With slight modification the proof works for any non-trivial scoring protocol for two candidates.

We reduce from the De-Negre (DNeg) variant of bi-level knapsack problem, which is proved to be $\Sigma_2^p$-hard by Caprara et al.~\cite{caprara2014study}. Before we describe the bi-level knapsack problem, we first introduce the classical knapsack problem, which is closely related. In the knapsack problem, given is some fixed budget $\bar{B}$ together with a set $S$ of items, each having a price $\bar{p}_j^a$ and a weight $\bar{w}_j$. The goal is to select a subset of items whose total price is no more than the given budget and the total weight is maximized. We denote by $KP(S,\bar{B})$ the optimal objective value of the knapsack problem.

In the De-Negre (DNeg) variant of bi-level knapsack problem, there is an adversary and a packer. The adversary has a reserving budget $\bar{F}$ and the packer has a packing budget $\bar{B}$. There is a set of $n$ items, each having a price $\bar{p}_j^a$ to the adversary, $\bar{p}_j^b$ to the packer and a weight $\bar{w}_j=\bar{p}_j^b$ (to both the adversary and the packer). The adversary first reserves a subset of items whose total prices is no more than $\bar{F}$. Then the packer solves the knapsack problem with respect to the remaining items that are not reserved, i.e., the packer will select a subset of remaining items whose total price is no more than $\bar{B}$ such that their total weight is maximized. The DNeg variant of the bi-level knapsack problem asks for a proper subset of items reserved by the adversary such that the total weight of items selected by the packer is minimized. More precisely, the problem can be formulated as a bi-level integer programming as follows.

\noindent\fbox{\begin{minipage}{0.97\textwidth}
		\textbf{The DNeg variant of bi-level knapsack problem:}
		\begin{subequations}
			\begin{align*}
			\text{Minimize }\hspace{1mm} & \sum_{j=1}^n \bar{p}_j^by_j&\\
			s.t. \hspace{1mm}&\sum_{j=1}^n \bar{p}_j^a x_j\le \bar{F}& \\
			& \text{where }  y_1,y_2,\cdots,y_n \text{ solves the following:}&\\
			&\text{Maximize}\quad  \sum_{j=1}^n\bar{p}_j^by_j&\\
			& s.t. \quad \sum_{j=1}^n\bar{p}_j^by_j\le \bar{B}&\\
			& \hspace{10mm} x_j+y_j\le 1\quad\quad 1\le j\le n&\\
			& \hspace{10mm} x_j,y_j\in\{0,1\}\quad 1\le j\le n&
			\end{align*}
		\end{subequations}
\end{minipage}}

The decision version of the DNeg variant of bi-level knapsack problem asks whether there exists a feasible solution with the objective value at most $W$. The following lemma is due to Caprara et al.~\cite{caprara2014study}.

\begin{lemma}[By Caprara et al. \cite{caprara2014study}]
	The decision version of the DNeg variant of bi-level knapsack problem is $\Sigma_2^p$-complete.
\end{lemma}
Based on the above lemma, we are able to prove Lemma~\ref{lemma:general-hardness}.
\begin{proof}[Proof of Lemma~\ref{lemma:general-hardness}]
	Given an arbitrary instance of the (decision version of) DNeg variant of bi-level knapsack problem, we construct an election instance as follows. There are $m=2$ candidates. The defense and attack budgets are $F=\bar{F}$ and $B=\bar{B}$, respectively. There are $n+2$ voters:
	\begin{itemize}
		\item $n$ key voters $v_1,v_2,\cdots,v_n$ who vote for $c_2$, each having an awarding price $p_j^a=\bar{p}_j^a$, a bribing price $p_j^b=\bar{p}_j^b$, and a weight $w_j=\bar{p}_j^b$.
		\item one dummy voter $v_{n+1}$ who votes for $c_2$ whose weight is $2W$, awarding price is $F+1$ and bribing price is $B+1$.
		\item one dummy voter $v_{n+2}$ who votes for $c_1$ whose weight is $\sum_{j=1}^n\bar{p}_j^b$, awarding price is $F+1$ and bribing price is $B+1$.
	\end{itemize}
	
	Obviously $c_2$ is the original winner. We show in the following that the constructed election instance is secure if and only if the DNeg variant of bi-level knapsack problem admits a feasible solution with an objective value at most $W$.
	
	Suppose the DNeg variant of bi-level knapsack problem admits a feasible solution with an objective value at most $W$, and let $x_i^*$ be such a solution. As $\sum_{j=1}^n\bar{p}_j^ax_j^*=\sum_{j=1}^n{p}_j^ax_j^*\le \bar{F}=F$, we let the defender award all the voters such that $x_j^*=1$, i.e., let $\VV_F=\{v_j|x_j^*=1\}$. According to the fact that the objective value of the bi-level knapsack problem is at most $W$, and the fact that the two dummy voters can never be bribed, it follows that the optimal objective value of the following knapsack problem is at most $W$:
	\begin{subequations}
		\begin{align*}
		\text{Maximize   }\hspace{10mm} &\sum_{j:v_j\in \VV\setminus\VV_F} {p}_j^by_j&\\
		\text{s.t. }\hspace{10mm} & \sum_{j\in \VV\setminus\VV_F} {p}_j^by_j\le {B}  &
		\end{align*}
	\end{subequations}
	Thus, with a budget of $B$ the the briber can never bribe key voters whose total weight is more than $W$. Note that originally the total weight of voters voting for $c_2$ is $2W+\sum_{j=1}^n p_j^b$, and the total weight of voters voting for $c_1$ is $\sum_{j=1}^n p_j^b$, the briber cannot succeed and the election is thus secure.
	
	Suppose the election is secure. Then there exists some $\VV_F$ such that we can say $\sum_{j:v_j\in\VV_F}  p_j^a\le F$ and if voters in $\VV_F$ are fixed, no briber can succeed. Note that the two dummy voters can never be protected nor bribed, whereas $\VV_F\subseteq \{v_1,v_2,\cdots,v_n\}$. Thus, among voters in $\{v_1,v_2,\cdots,v_n\}\setminus \VV_F$, within a budget of $B$ the briber cannot bribe voters whose total weight is more than $W$. This is equivalent as saying that by setting $x_j=1$ for $v_j\in \VV_F$ and $x_j=0$ otherwise, the knapsack problem for $y_j$ does not admit a feasible solution with an objective value more than $W$. Hence, the objective value of the given DNeg variant of bi-level knapsack problem is at most $W$.
	\qed
\end{proof}

\subsection{Proofs Omitted in Section~\ref{subsec:weighted-unit}}


\begingroup
\def\thetheorem{\ref{thm:uni-price-plu-p}}
\begin{theorem}
	If $m$ is a constant, then the destructive weighted-protection problem is in \PP\, for any scoring rule.
\end{theorem}
\addtocounter{theorem}{-1}
\endgroup

\begin{proof}
The theorem is proved by trying all different possible $\VV_F$ and check whether the attacker can succeed for each of them. 	
	By Corollary~\ref{coro:maximal}, for voters having the same preference, $\VV_F$ contains voters of the largest weights. Hence to determine $\VV_F$, it suffices to know the number of voters having each preference in $\VV_F$. There are at most $m!$ different preferences, and consequently at most $n^{m!}$ different kinds of $\VV_F$, which is polynomial when $m$ is constant. For each possible choice of $\VV_F$, we check whether the attacker can succeed by trying all possible $\VV_B$ and all possible ways of changing their preferences. Firstly, by Corollary~\ref{coro:maximal}, for voters in $\VV\setminus\VV_F$ that have the same preference list, $\VV_B$ contains the ones of the largest weights, hence using a similar argument we know there are at most $n^{m!}$ different kinds of $\VV_B$. Given $\VV_F$ and $\VV_B$, it remains to determine how the preference lists of voters in $\VV_B$ should be changed. Note that we do not need to specify how the preference list is changed for each $v_j\in \VV_B$. Instead, we only need to determine the number of voters in $\VV_B$ that are changed to each preference list, which gives rise to at most $n^{m!}$ possibilities. Therefore, overall there are at most $n^{3m!}$ different possibilities regarding $\VV_F$, $\VV_B$ and how to alter the preference lists of voters in $\VV_B$, which can be enumerated efficiently when $m$ is a constant.
	\qed
\end{proof}
{We remark that an argument similar to the one we used to prove Theorem \ref{thm:uni-price-plu-p} was used by Faliszewsk et al. \cite{faliszewski2009hard}.}

\subsection{Proofs omitted in Section~\ref{subsec:unit-dollar}}

\begingroup
\def\thetheorem{\ref{thm:assy-uniweight-np}}
\begin{theorem}
	For any non-trivial scoring rule, both the constructive and destructive \$-protection problems are \NP-complete.
\end{theorem}
\addtocounter{theorem}{-1}
\endgroup

To prove Theorem \ref{thm:assy-uniweight-np}, we first show the problem belongs to \NP\ in Lemma \ref{lemma:uni-weight-NP-easy} and then show its NP-hardness in Lemma \ref{lemma:unit-weight-hardness}.

\begin{lemma}
\label{lemma:uni-weight-NP-easy}
For any scoring rule and arbitrary constant $m$, both the constructive and destructive \$-protection problems
are in \NP.
\end{lemma}
\begin{proof}
	Note that to show the membership in \NP, it suffices to show that given $\VV_F$, we can determine in polynomial time whether the constructive/destructive attacker can succeed. Note that among voters of the same preference list in $\VV\setminus\VV_F$, $\VV_B$ always contains the ones with the smallest bribing prices. 
	 Hence, similarly as the proof of Theorem~\ref{thm:uni-price-plu-p}, there are at most $n^{m!}$ different kinds of $\VV_B$. Given $\VV_B$, a similar argument as that of Theorem~\ref{thm:uni-price-plu-p} shows that there are $n^{m!}$ different ways of altering the preference lists of voters in $\VV_B$. Hence in $n^{2m!}$ time we can determine the whether the constructive/destructive attacker can succeed, which is polynomial if $m$ is a constant.
	 \qed
 \end{proof}
\begin{lemma}
\label{lemma:unit-weight-hardness}
For any non-trivial scoring rule, both the constructive and destructive \$-protection problems
are \NP-hard even if there are only 2 candidates.	
\end{lemma}

Again, in case of two candidates, the constructive and destructive variants are identical and it suffices to prove the theorem under the scoring rule of plurality.

Towards the proof, we need the following intermediate problems.

\noindent\textbf{Balanced Partition:} Given a set of positive integers $\{a_1,a_2,\cdots,a_{2n}\}$ where $a_1\le a_2\le \cdots\le a_{2n}$ and an integer $q$ such that $\sum_{j=1}^{2n}a_j=2q$. Determine whether there exists a subset $S$ of $n$ integers such that $\sum_{i:a_i\in S}a_i=q$.

The balanced partition problem is a variant of the partition problem (in which $S$ is not required to contain exactly $n$ integers). The \NP-completeness of the balanced partition problem is a folklore result, which follows from a slight modification on \NP-completeness proof for the partition problem~ given by Garey and Johnson\cite{garey2002computers}.

Using the balanced partition problem, we are able to show the \NP-hardness of the following problem in Lemma~\ref{lemma:balanced}.

\smallskip
\noindent\textbf{Balanced partition$'$:} Given a set of positive integers $\{a_1,a_2,\cdots,a_{2n}\}$ where $a_1\le a_2\le \cdots\le a_{2n}$ and an integer $q$ such that $\sum_{j=1}^{2n}a_j=2q$, $a_{n}+a_{n+1}+\cdots+a_{2n-1}\ge q+1$. Determine whether there exists a subset $S$ of $n$ integers such that $\sum_{i:a_i\in S}a_i=q$.

\begin{lemma}\label{lemma:balanced}
	Balanced partition$'$ is \NP-complete.
\end{lemma}
\begin{proof}
	Membership in \NP\, is straightforward. We show balanced partition$'$ is \NP-hard in the following via reduction from balanced partition. 

	Given an instance of the balanced partition problem where the integers are $a_1\le a_2\le \cdots\le a_{2n}$ and $q=1/2\cdot \sum_{j=1}^{2n}a_j$, we construct an instance of the balanced partition$'$ problem by adding $4n$ integers, each of value $3q$.
	
		We first show that the constructed instance is a feasible instance. Obviously every additional integer is larger than any $a_j$ where $j\le 2n$. Let the additional integers be $a_{2n+1},a_{2n+2},\cdots,a_{6n}$. In the constructed instance there are $2n'=6n$ elements, with the summation of all integers being $2q+12nq$. Let $q'=q+6nq$. Obviously $a_{n'+1}+a_{n'+2}+\cdots+a_{2n'-1}=9nq>q'+1$. Thus, the constructed instance is a feasible instance of the balanced partition$'$ problem.
	
	We show that the constructed instance admits a feasible partition if and only if the given balanced partition instance admits a feasible solution.
	
	If the given balanced partition instance admits a feasible solution, then obviously the constructed balanced partition$'$ problem admits a feasible solution (by adding $2n$ of the additional integers to both sides).
	
	Suppose the constructed balanced partition$'$ instance admits a feasible solution. We claim that among the $4n$ additional integers, there are exactly $2n$ of them in $S$. Otherwise $S$ contains either at most $2n-1$ of them or at least $2n+1$ of them. By symmetry we assume without loss of generality that $S$ contains at most $2n-1$ of them, then all integers in $S$ add up to at most $(2n-1)\cdot 3q+2q<q'=q+6nq$, which is a contradiction. Thus, $S$ contains exactly $2n$ additional integers, implying that the remaining $n$ integers adding up to $q$, i.e., the given balanced partition instance admits a feasible solution.
	\qed
\end{proof}
\begin{proof}[Proof of Lemma~\ref{lemma:unit-weight-hardness}]
	We will prove the NP-hardness for an election with only two candidates under $1$-approval (plurality). Recall that in this case the constructive and destructive protection problem are the same.
	
	We reduce from the balanced partition$'$ problem. Given an arbitrary instance of the balanced partition$'$ problem, we construct an instance of the \$-bribery-protection problem such that the answer to the problem is ``Yes" if and only the balanced partition$'$ instance admits a feasible solution.
	
	We construct the \$-bribery-protection instance as follows. There are $m=2$ candidates. $c_1$ is the designated candidate. Let $F=4qn+q$, $B=(4n-1)q-1$. There are $6n-1$ voters, each of unit weight and can be divided into three groups:
	\begin{itemize}
		\item $2n$ key voters voting for $c_2$, whose awarding prices are $4q+a_1,4q+a_2,\cdots,4q+a_{2n}$ and bribing prices are $4q-a_1,4q-a_2,\cdots,4q-a_{2n}$, respectively. 
		\item $2n-1$ dummy voters voting for $c_2$, whose awarding prices are all $F+1$ and bribing prices are all $B+1$.
		\item $2n$ dummy voters voting for $c_1$, whose awarding and bribing prices are all $1$.
	\end{itemize}
	Let $\VV^*=\{v_1,v_2,\cdots,v_n\}$ be the set of key voters.
	
	Obviously $c_2$ is the original winner. We show that the answer to the constructed \$-bribery-protection instance is ``Yes" if and only if the balanced partition$'$ problem admits a feasible solution.
	
	Suppose the given balanced partition$'$ instance admits a feasible solution $S$. Now we let the defender fix the $n$ key voters whose awarding price is $4q+a_j$ where $a_j\in S$. It is easy to verify that the total awarding price is $4nq+q=F$, which stays within the defense budget. We argue that, no briber can alter the election result with a budget of $B$. Let $\VV_F$ be the set of fixed voters. Suppose on the contrary there is a briber who can make $c_1$ win. Given that the total attack budget is $B$, the briber can only bribe key voters. Furthermore, the briber has to bribe at least $n$ voters. As $|\VV_F|=n$, the briber has to bribe all voters in $\VV^*\setminus \VV_F$. However,
	$$\sum_{j:v_j\in \VV^*\setminus \VV_F} p_j^b=(4n-1)q>B,$$ which is a contradiction. Thus, the answer to the constructed \$-bribery-protection instance is ``Yes".
	
	Suppose the answer to the constructed \$-bribery-protection instance is ``Yes". Note that in total there are $4n-1$ voters voting for $c_2$ and $2n$ voters voting for $c_1$. Thus any briber who wants to alter the election result has to bribe at least $n$ voters who originally vote for $c_2$. Since the $2n-1$ dummy jobs voting for $c_2$ can never be protected nor bribed, we can fix a subset $\VV_F\subseteq \VV^*$ such that no briber can bribe $n$ or more voters from $\VV^*\setminus\VV_F$ with a budget of $B$. We have the following claim.
	\begin{claim}\label{claim:=k}
		$|\VV_F|=n$.
	\end{claim}
	\begin{proof}[Proof of Claim~\ref{claim:=k}]
		We first show that $|\VV_F|\le n$. Suppose on the contrary that $|\VV_F|\ge n+1$. Note that any key voter has an awarding price of at least $4q$, thus the total awarding price is at least $4(n+1)q$. However, $F=4qn+q<4(n+1)q$, which is a contradiction.
		
		We now show that $|\VV_F|\ge n$. Suppose on the contrary that $|\VV_F|\le n-1$. Then there are at least $n+1$ key voters that can be bribed and we bribe the cheapest $n$ voters. As $p_1^b\ge p_2^b\ge\cdots\ge p_{2n}^b$, the total bribing price of the cheapest $n$ voters among any $n+1$ voters is at most $p_{n}^b+p_{n+1}^b+\cdots+p_{2n-1}^b=4qn-(a_n+a_{n+1}+\cdots+a_{2n-1})\le 4qn-(q+1)=B$, whereas the briber can always bribe the cheapest $n$ voters and let $c_2$ win, which contradicts the fact that the answer to the \$-bribery-protection instance is ``Yes". Hence $|\VV_F|\ge n$.
		\qed
	\end{proof}
	
	Now the following inequalities hold simultaneously:
	
	\begin{subequations}
		\begin{align}
		&\quad \sum_{j:v_j\in\VV_F} p_j^a\le F=4qn+q~\label{eq:1a}\\
		&\quad \sum_{j:v_i\in \VV^*\setminus\VV_F} p_j^b\ge B+1=4(n-1/2)q+q~\label{eq:1b}
		\end{align}	
	\end{subequations}
	Note that
	\begin{eqnarray*}
		&\sum_{j:v_i\in \VV^*\setminus\VV_F} p_j^b=\sum_{j=1}^{2n}p_j^b-\sum_{j:v_j\in\VV_F}p_j^b\\=&4(2n-1/2)q-\sum_{j:v_j\in\VV_F}(4q-a_j),
	\end{eqnarray*}
	
	by~(\ref{eq:1a}) we have
	$$\sum_{j:v_j\in\VV_F}(4q-a_j)\le 4qn-q.$$
	Using the fact that $|\VV_F|=n$, we have
	$$\sum_{j:v_j\in\VV_F}a_j\ge q.$$
	From (\ref{eq:1b}), we have
	$$\sum_{j:v_j\in\VV_F}a_j\le q.$$
	Thus,
	$$\sum_{j:v_j\in\VV_F}a_j= q,$$
	i.e., the given balanced partition$'$ instance admits a feasible solution. 	
	\qed
\end{proof}

The symmetric \$-protection problem (i.e., $p_j^a=p_j^b$), however, is  significantly easier, as shown by Theorem \ref{thm:uni-weight-plu-p}.


\begingroup
\def\thetheorem{\ref{thm:uni-weight-plu-p}}
\begin{theorem}
For constant $m$, both destructive and
	constructive symmetric \$-protection problems are in $P$ for any scoring rule.
\end{theorem}
\addtocounter{theorem}{-1}
\endgroup
\begin{proof}
{The proof idea is the same as that of Theorem~\ref{thm:uni-price-plu-p}, namely by trying all different possible $\VV_F$. For each $\VV_F$, we try all possible $\VV_B$'s and all possible ways of altering the preference of voters in $\VV_B$. Note that every voter has the same weight and satisfies that $p_j^a=p_j^b$. Therefore, a voter with a smaller (awarding and bribing) price always dominates a voter with a larger price. By Corollary~\ref{coro:maximal} and Corollary~\ref{coro:maximal2}, for the voters having the same preference, $\VV_F$ contains the voters that have the smallest {prices}. Therefore, in order to determine $\VV_F$, it suffices to know the number of voters that have the same preference, implying that there are at most $n^{m!}$ different kinds of {$\VV_F$'s.} Similarly, given {a $\VV_F$, there are at most $n^{m!}$ different kinds of $\VV_B$'s.} Using the same argument as that of Theorem~\ref{thm:uni-price-plu-p}, for every $\VV_F$ and $\VV_B$, there are at most $n^{m!}$ ways of altering the preferences of the voters in $\VV_B$. Therefore, there are $n^{3m!}$ different possibilities {in total}, which can be enumerated efficiently when $m$ is a constant.}
\qed
\end{proof}

\subsection{Proofs Omitted in Section~\ref{subsec: destructive-arbitrary}}\label{appsec:des}
The goal is to prove Theorem~\ref{thm:destructive-np}.


\begingroup
\def\thetheorem{\ref{thm:destructive-np}}
\begin{theorem}
	Both $r$-approval destructive weighted-protection and $r$-approval (symmetric) \$-protection problems are \NP-complete.
\end{theorem}
\addtocounter{theorem}{-1}
\endgroup

Towards the proof, we first show the \NP-hardness.

\begin{lemma}\label{lemma:destructive-np}
The $r$-approval destructive weighted-protection problem is \NP-hard for any $r\ge 3$.
\end{lemma}
\begin{proof}
	We prove the lemma for $r=3$. The case of $r>3$ can be proved by introducing dummy candidates and letting each voter vote for $r-3$ distinct dummy candidate. 
	
	We reduce from a variant of 3-dimensional matching in which every element occurs at most $d=O(1)$ times in the given triples, which is also known to be \NP-hard~ stated by Kann \cite{kann1991maximum}.
	Given a 3DM instance with $3\zeta=|W\cup X\cup Y|$ elements and $\eta=|M|$ triples such that every element appears at most $d=O(1)$ times in $M$, we construct an instance of the destructive weighted bribery-protection problem as follows. Here we further require that $\eta\ge \zeta+2d$. The assumption is without loss of generality since if $\eta\le \zeta+2d-1$, then there are at most $O(1)$ triples outside a perfect matching, and the existence of a perfect matching can be determined by brute-forcing within $\zeta^{O(1)}$ time, which is polynomial.
	
	For ease of description, we re-index all elements of $W\cup X\cup Y$ arbitrarily as $z_1,z_2,\cdots,z_{3\zeta}$.
	
	Let $Q=2\eta+1$. There are $3\zeta+1$ key candidates, including the following two kinds of candidates (the function $f$ will be defined later):
	\begin{itemize}
		\item $3\zeta$ key candidates $c_1$ to $c_{3\zeta}$, with $c_i$ corresponding to $z_i\in W\cup X\cup Y$ and has a score of $Q\cdot f(z_i)$. We call them element candidates;
		\item one key candidate $c_{3\zeta+1}$ called leading candidate, which is the original winner and has a score of $Q\cdot f(z_{3\zeta+1})$.
	\end{itemize}
	Besides key candidates, there are also sufficiently many dummy candidates $c_{i}$ for $i>3\zeta+1$, each having a score of $1$. The number of dummy candidates will be determined later.
	
	
	There are $\eta$ key voters $v_1$ to $v_\eta$, each of weight $Q$. Each key voter corresponds to a distinct triple in $(z_i,z_j,z_k)\in M$, and votes for the three candidates that correspond to $z_i,z_j,z_k$, respectively.
	
	Besides key voters, there are also sufficiently many dummy voters $v_j$ for $j>\eta$. A dummy vote has a unit weight, and votes for one key candidate and two distinct dummy candidates.
	
	Now we determine the number of dummy voters and dummy candidates together with all the parameters.
	If we only consider key voters of weight $Q$, then every element candidate corresponding to some $z_i$ gets a score of $Q\cdot d(z_i)$ where $d(z_i)$ is the number of occurrences of $z$ in $M$. Adopting the viewpoint of the minmax vector addition problem, for every $1\le j\le \eta$ we have
	\begin{equation*}
	\Delta_{ij}=\begin{cases}
	0, \hspace{5mm} &\text{if $v_j$ votes for $c_i$ and does not vote for $c_{3\zeta+1}$}\\
	1,  &\text{if $v_j$ votes for $c_{3\zeta+1}$ and does not votes for $c_i$}, \\ & \text{or $v_j$ votes for both $c_i$ and $c_{3\zeta+1}$}\\
	2, & \text{if $v_j$ votes for $c_{3\zeta+1}$ and does not vote for $c_i$}
	\end{cases}
	\end{equation*}
	Let $$\Delta_{max}=\max_{1\le i\le 3\zeta}\sum_{j=1}^\zeta\Delta_{ij},\quad d_{max}=\max_{1\le i\le 3\zeta} d(z_i).$$ We define $$f(z_i)=2\eta+ d_{max}+\Delta_{max}-\sum_{j=1}^{\eta}\Delta_{ij}, 1\le i\le 3\zeta$$
	That means, candidate $c_i$, $1\le i\le 3\zeta$ will get a score of $Q\cdot d(z_i)$ from key voters, and additionally $Q\cdot [f(z_i)-d(z_i)]\ge 0$ score from dummy voters. Hence, for each $c_i$ we need to create $Q\cdot [f(z_i)-d(z_i)]$ dummy voters.
	
	We define $$f(z_{3\zeta+1})=2\eta+d_{max}+\Delta_{max}-\zeta+1.$$ Note that $\sum_{j=1}^{\eta}\Delta_{ij}\ge \eta-d>\zeta$ as every element appears at most $d$ times in triples, hence $f(z_{3\zeta+1})>f(z_i)$ for $1\le i\le 3\zeta$ and $c_{3\zeta+1}$ is indeed the original winner.
	
	Overall, dummy voters should contribute $Q\cdot [f(z_i)-d(z_i)]$ to each $c_i$, $1\le i\le 3\zeta$ and $(\zeta+1)f(z_{3\zeta+1})$ to $c_{3\zeta+1}$. We create in total $Q\cdot [\sum_{i=1}^{3\zeta+1}f(z_i)-\sum_{i=1}^{3\zeta}d(z_i)]$ dummy voters, and $2Q\cdot [\sum_{i=1}^{3\zeta+1}f(z_i)-\sum_{i=1}^{3\zeta}d(z_i)]$ dummy candidates.
	
	Let the defense budget be $F=\zeta$ and the attack budget be $B=\eta-\zeta$.
	
	\smallskip
	\noindent\textbf{``Yes" Instance of 3DM $\to$ ``Yes" Instance of Destructive Weighted-Bribery-Protection.} Suppose the given 3DM instance admits a feasible solution, we show that the answer to destructive weighted-bribery-protection problem is ``Yes". Let $T\subseteq M$ be the perfect matching. Then $|T|=\zeta$ and we let the defender protect voters corresponding to the triples in $T$. Taking the viewpoint of the minmax vector addition problem. If the attacker bribes all the key voters, then $W(c_i)$ increases by exactly $Q\cdot\sum_{j=1}^{\eta}\Delta_{ij}$ for $1\le i\le 3\zeta$. First it is easy to see that no dummy candidate can be a winner as $Q\cdot\sum_{j=1}^{\eta}\Delta_{ij}\le 2{\eta}Q$ while $Q\cdot f(z_{3\zeta+1})\ge (2\eta+1)Q$. Meanwhile, for each key candidate $c_{i}$, $1\le i\le 3\zeta$, his/her total score becomes exactly $Q\cdot [f(z_{3\zeta+1})+\zeta-1]$. As the defender fixes a subset of key voters, we should subtract the contribution of these key voters. As the triple corresponding to these voters form a perfect matching, these voters contribute a score of exactly $Q(\zeta-1)$ to each $c_i$, hence after bribery every key candidate has a score at most $Q\cdot f(z_{3\zeta+1})$, implying that the answer to the Destructive Weighted-Bribery-Protection problem is ``Yes".
	
	\smallskip
	\noindent\textbf{``No" Instance of 3DM $\to$ ``No" Instance of Destructive Weighted-Bribery-Protection.} Suppose the given 3DM instance does not admit a perfect matching, we show that the answer to the constructed instance of the destructive weighted-bribery-protection problem is ``No". Consider an arbitrary set of voters fixed by the defender and let $U$ be the subset of key voters that are fixed. Obviously $|U|\le \zeta$. If $|U|<\zeta$, we add arbitrary key voters into $U$ such that its cardinality becomes $\zeta$. Let $U'$ be the set of these $\zeta$ key voters and let the attacker bribe the remaining $\zeta-\eta$ key voters. Again we take the viewpoint of the minmax vector addition problem. If the attacker bribes every key voter, then the total score of every key candidate $c_i$, $1\le i\le 3\zeta$, becomes exactly $Q\cdot [f(z_{3\zeta+1})+\zeta-1]$. As key voters in $U$ are not bribed, we subtract their contribution from each $c_i$. Note that triples corresponding to voters in $U$ do not form a perfect matching, thus there exists some element which appears at least twice in these triples. Let $z_{k}$ be such element and we consider $c_k$. It is clear that $\Delta_{kj}=0$ if $v_j$ votes for $c_k$ and $\Delta_{kj}=1$ if $v_j$ does not vote for $c_k$ (note that key voters never vote for $c_{3\zeta+1}$). Hence $\sum_{j:v_j\in U} \Delta_{kj}\le \zeta-2$. By subtracting the contribution of voters in $U$, $c_k$ has a score at least $Q\cdot [f(z_{3\zeta+1}+1)]$, implying that after bribery $c_k$ will get a higher score than $c_{3\zeta+1}$. Thus, the answer to the Destructive Weighted-Bribery-Protection problem is ``No".	
	\qed
\end{proof}


Note that in the preceding reduction, we only construct voters of two different weights, $Q$ for the key voters and $1$ for the dummy voters. {Recall that $Q$ is set to be large enough to assure that only the key voters will be considered by the defender or the attacker.} Once $\VV_F$ and $\VV_B$ are restricted to be subsets of the key voters, the concrete value of $Q$ does not matter. Moreover, we can also prove the \NP-hardness of the destructive (symmetric) \$-bribery-protection problem by using essentially the same proof, except that we set key voters of price $1$ and dummy voters of price exceeding budgets $F$ and $B$, say, $\max\{F,B\}+1$. This leads to the following lemma.

\begin{lemma}\label{lemma: destructive (symmetric)}
The $r$-approval destructive (symmetric) \$-bribery-protection problem is \NP-hard for any $r\ge 3$.
\end{lemma}

Having showed the \NP-hardness of the destructive weighted-protection problem, we show the problem is polynomial-time verifiable and is therefore {\NP}-complete.

\begin{lemma}\label{lemma:destructive-veri-p}
The destructive weighted-protection problem can be verified in polynomial time under any scoring rule.
\end{lemma}

\begin{proof}
We leverage the {\em minmax vector addition} problem.
In the case of unit price, given $\VV_F$, the decision version of the verification problem becomes: does there exist a subset $\VV_B\subseteq\VV\setminus\VV_F$ such that the following is true
\begin{eqnarray*}
\left\|\vec{\Lambda}+\sum_{j:v_j\in \VV_B}\vec{\Delta}_j\right\|>\Lambda(c_m).
\end{eqnarray*}
To answer this decision problem, it suffices to do the following for every $1\le i\le m-1$: pick $B$ voters from $\VV\setminus\VV_F$ whose $i$-th coordinate $\Delta_{ij}$ is the largest, add them to $\Lambda(c_i)$, and check if it is greater than $\Lambda(c_m)$.
\qed
\end{proof}

It is, however, not clear if the destructive \$-protection problem is {\NP}-complete for arbitrary scoring rules. However, we show in the following that for any scoring rule which only assigns a constant number of different scores to a preference list, i.e., the $\alpha_i$'s only take $O(1)$ distinct values, the \$-protection problem can be verified in polynomial-time. As in the case of the $r$-approval rule, the $\alpha_i$'s only take values of $1$ or $0$, the destructive \$-protection problem is {\NP}-complete for the $r$-approval rule.

\begin{lemma}\label{lemma:destructive-veri-price-p}
The destructive (symmetric) \$-protection problem can be verified in polynomial-time in $n$ under any scoring rule in which the $\alpha_i$'s only take a constant number of distinct values.
\end{lemma}
\begin{proof}
Consider the {\em minmax vector addition} problem. We observe that in the case of unit weight and that the $\alpha_i$'s take $O(1)$ distinct values, $\Delta_{ij}$ only takes $O(1)$ distinct values. For each coordinate $i$, we can check if it is possible for $\Lambda(c_i)$ and $\Delta_{ij}$'s to add up to some value strictly greater than $\Lambda(c_m)$. Note that by adding every $\Delta_{ij}$, we need to pay a price of $p_j^b$, hence it is essentially the {\em knapsack} problem with items having arbitrary prices but only $O(1)$ distinct weights. Such a knapsack problem can be solved in polynomial-time, e.g., by simply guessing the number of items of the same weight. Among the items of the same weight, the optimal solution should take the ones with the cheapest price.
\qed
\end{proof}

\clearpage

\footnotesize
\bibliographystyle{splncs04}
\bibliography{bibliography}

\end{document}